\documentclass[12pt]{article}
\usepackage{amsmath}
\usepackage{amssymb}
\usepackage{amsthm}
\usepackage{mathtools}
\usepackage{url}
\usepackage[usenames,svgnames]{xcolor}
\usepackage{enumitem}
\usepackage{comment}
\usepackage{mathdots}
\usepackage{graphicx}
\usepackage{float}
\usepackage{enumitem}

\usepackage[pagebackref=false]{hyperref} 
\usepackage[all]{hypcap} 

\theoremstyle{plain}
\newtheorem{theorem}{Theorem}[section]

\newtheorem{definition-theorem}[theorem]{Definition-Theorem}
\newtheorem{definition-proposition}[theorem]{Definition-Proposition}

\newtheorem{proposition}[theorem]{Proposition}

\newtheorem{example}{Example}[section]
\newtheorem{examples}{Example}[subsection]

\newtheorem{remark}{Remark}[section]

\theoremstyle{definition}
\newtheorem{definition}{Definition}[section]

\numberwithin{equation}{section} 

 \newcommand{\sgn}{\mathop{\mathrm{sgn}}}

\DeclareMathOperator{\cyc}{cyc}

\DeclareMathOperator{\aut}{aut}

\def\ra{{\rightarrow}}

\def\wt{\widetilde}

\def\be{\begin{equation}}
\def\ee{\end{equation}}

\def\bea{\begin{eqnarray}}
\def\eea{\end{eqnarray}}

\def\bt{\begin{theorem}}
\def\et{\end{theorem}}

\def\bex{\begin{example}\small \rm}
\def\eex{\end{example}}

\def\bexs{\begin{examples}\small \rm}
\def\eexs{\end{examples}}

\def\ra{\rightarrow}
\def\ss{\subset}

\def\br{\begin{remark}\small \rm}
\def\er{\end{remark}}

\def\wt{\widetilde}

\def\&{&{\hskip -20pt}}

\def\CC{\mathcal{C}}

\def\MM{\mathcal{M}}

\def\WW{\mathcal{W}}

\def\Ib{\mathbf{I}}

\def\Nb{\mathbf{N}}

\def\Nb{\mathbf{N}}
\def\Pb{\mathbf{P}}

\def\Zb{\mathbf{Z}}

\begin{document}
\baselineskip 16pt
\medskip
\begin{center}
\begin{Large}\fontfamily{cmss}
\fontsize{17pt}{27pt}
\selectfont
	\textbf{Generating weighted Hurwitz numbers}
	\end{Large}
	
\bigskip \bigskip
\begin{large}  M. Bertola$^{1, 2, 3}$\footnote{e-mail: Marco.Bertola@concordia.ca, Marco.Bertola@sissa.it}, 
J. Harnad$^{1, 2}$\footnote {e-mail: harnad@crm.umontreal.ca}  
and B. Runov$^{1,2}$\footnote{e-mail: boris.runov@concordia.ca}
 \end{large}
 \\
\bigskip
\begin{small}
$^{1}${\em Department of Mathematics and Statistics, Concordia University\\ 1455 de Maisonneuve Blvd.~W.~Montreal, QC H3G 1M8  Canada}\\
\smallskip
$^{2}${\em Centre de recherches math\'ematiques, Universit\'e de Montr\'eal, \\C.~P.~6128, succ. centre ville, Montr\'eal, QC H3C 3J7  Canada}\\
$^{3}${\em SISSA/ISAS, via Bonomea 265, Trieste, Italy }
\end{small}
 \end{center}
\medskip
\begin{abstract}
    Multicurrent correlators associated to KP $\tau$-functions of hypergeometric
type are used as generating functions for weighted Hurwitz numbers. 
These are expressed as formal Taylor series  and used to compute  generic, simple, rational and quantum weighted single Hurwitz numbers.
\end{abstract}

\section{Introduction}

It is well-known that KP  $\tau$-functions of hypergeometric type may serve as combinatorial generating functions for 
weighted Hurwitz numbers \cite{Ok, Pa,  AMMN, GH1, HO, GH2, H1}. An efficient way of computing the latter is
to make use of the associated multicurrent correlators \cite{ACEH2, ACEH3}.  This method 
is implemented in the following, without recourse to matrix integral representations \cite{AC1, AC2, BH}
or topological recursion \cite{EO2, BEMS, KZ, ACEH1, ACEH2, ACEH3}.

Section \ref{gen_fns_weighted_Hurwitz} reviews  the use of KP and $2D$-Toda $\tau$-functions
as  generating functions for weighted Hurwitz numbers.
Their definition is recalled, with the special cases of {\em simple} Hurwitz numbers, {\em rationally weighted} 
Hurwitz numbers (\ref{G_ratl_c_d}) and {\em quantum} Hurwitz numbers as illustrative examples. 
 The starting point, developed in refs.~\cite{GH1, HO, GH2}, is Theorem \ref{tau_gener_fn_weighted_hurwitz}, 
 which, generalizing the earlier results of Okounkov \cite{Ok} and Pandharipande \cite{Pa} for simple Hurwitz numbers,
  states that KP $\tau$-functions \cite{Sa, SW}  of special {\em hypergeometric} type \cite{OrSc} are generating 
  functions for weighted Hurwitz numbers corresponding to arbitrary weight generating functions $G(z)$,  
 
 The new element in the present work is the observation that a more effective tool is provided by the
  multicurrent correlators $W_g(x_1, \dots, x_n)$  introduced in  \cite{ACEH2, ACEH3} in the context of the 
  topological recursion (TR) approach. As explained in Proposition \ref{F_n_gen_fn}, these provide 
generating functions for single weighted Hurwitz numbers $H^d_G(\mu)$, grouped according to the length $n=\ell(\mu)$ 
of the partition $\mu$ giving the ramification profile over a selected branch point. Making use of the polynomial expressions 
for  $W_g(x_1, \dots, x_n)$ in terms of pair correlators $K(x_i, x_j)$ given in Proposition \ref{prop:detConnected},  together
with the explicit expressions (\ref{K0_exp}), (\ref{rho_ab}) of the latter as weighted polynomials in the Taylor 
coefficients $\{g_i\}_{i\in \Nb}$ of $G(z)$, the multicurrent correlators are directly expressible as series 
with coefficients that are  weighted polynomials in the $g_i$'s, without the need to solve any recursion relations. 

 The main result is Theorem \ref{weighted_hurwitz_formulae},  which provides explicit formulae
 for weighted Hurwitz numbers as graded polynomials in these Taylor coefficients.
  Using the multicurrent correlators as generating functions, the weighted single Hurwitz numbers  
  $H^d_G(\mu)$ are determined for ramification profile $\mu$ of length $n:= \ell(\mu)$ equal to $1, 2$ or $3$. 
  This is first done in the generic case, for arbitrary weight generating function $G(z)$, with explicit formulae  for  
  $H^d_G(\mu)$ for small values of $N$, $d$ and $\ell(\mu)$ given in Appendix \ref{eval_weighted_hurwitz}, 
  Tables \ref{Hurwitz_table4} - \ref{Hurwitz_table7}. These are  specialized to various particular cases in 
  Tables \ref{Hurwitz_table8}-\ref{Hurwitz_table13}: the exponential case $G(z)=e^z$, which determines 
{\em simple} Hurwitz numbers; polynomial and rational weight generating functions, which include, in particular, 
strictly and weakly monotone Hurwitz numbers as well as weighted enumeration of Belyi curves \cite{AC1,  HO,  Z, GGN1, GH2}, 
and, finally, quantum Hurwitz numbers \cite{H2, GH2}.

\section{Generating functions for weighted Hurwitz numbers}
\label{gen_fns_weighted_Hurwitz}

\subsection{Pure and weighted Hurwitz numbers}

We recall the definition of {\em pure} Hurwitz numbers \cite{Hu1, Frob1,  Hu2, Frob2, Sch, LZ} and {\em weighted} 
Hurwitz numbers \cite{GH1, HO, H1, GH2}. 
\begin{definition}[Combinatorial]
For a  set of $k$ partitions $\{\mu^{(i)}\}_{i=1,\dots, k}$ of $N\in \Nb^+$, the {\em pure} Hurwitz number 
$H(\mu^{(1)}, \dots, \mu^{(k)})$  is ${1\over N!}$ times the number of distinct ways that the identity element $\Ib_N \in S_N$ 
in the symmetric group on $N$ elements can be expressed as a product
\be
\Ib_N = h_1 \cdots h_k
\label{factoriz_hi}
\ee
of $k$ elements $\{h_i\in S_N\}_{i=1, \dots, k}$, such that for each $i$, $h_i$ belongs
to the conjugacy class $\cyc(\mu^{(i)})$ whose cycle lengths are equal to the parts of $\mu^{(i)}$:
\be
h_i \in \cyc(\mu^{(i)}), \quad i =1, \dots, k.
\ee
\end{definition}

An equivalent definition involves the enumeration of branched coverings of the Riemann sphere.
\begin{definition}[Geometric]
For a set of partitions $\{\mu^{(i)} \}_{i=1,\dots, k}$ of weight $|\mu^{(i)}|=N$, 
the pure Hurwitz number  $H(\mu^{(1)}, \dots, \mu^{(k)})$ is  defined geometrically  \cite{Hu1, Hu2}  as the number
of inequivalent $N$-fold branched coverings  $\CC \ra \Pb^1$  of the Riemann sphere with $k$ branch points 
whose ramification profiles are given by the partitions $\{\mu^{(1)}, \dots, \mu^{(k)}\}$, 
normalized by the inverse  $1/|\aut (\CC)|$ of the order of the automorphism group of the covering. 
 The connected pure Hurwitz number $\wt{H}(\mu^{(1)}, \dots, \mu^{(k)})$ is defined in the same way, 
 with only connected coverings included in the enumeration.
 \end{definition}
The equivalence of the two follows from  the monodromy homomorphism 
\be
\MM:\pi_1 (\Pb^1/\{Q^{(1)}, \dots, Q^{(k)})\} \ra S_N
\ee
from the fundamental group of the Riemann sphere punctured at the branch points into $S_N$
 obtained by lifting closed loops from the base to the covering.

 The Frobenius-Schur formula determines Hurwitz numbers $H(\mu^{(1)}, \dots, \mu^{(k)})$  in terms of
 irreducible characters of the symmetric group $S_N$
 \be
 H(\mu^{(1)}, \dots, \mu^{(k)}) = \sum_{\lambda, \ |\lambda|=N} (h(\lambda))^{k-2} \prod_{i=1}^k {\chi_\lambda(\mu^{(i)}) \over z_{\mu^{(i)}}},
 \label{Frob_Schur_Hurwitz}
 \ee
 where $\chi_\lambda(\mu)$ is the character of the irreducible representation 
 corresponding to the partition $\lambda$ of $N$ evaluated on the conjugacy class
 $\cyc(\mu) \ss S_N$ with cycle structure given by the partition $\mu$, $h(\lambda)$ is the product of 
 hook lengths of $\lambda$ and 
 \be
 z_\mu = \prod_{i=1}^{\mu_1} m_i(\mu))! i^{m_i(\mu)}
 \label{z_mu_def}
 \ee
 is the order of the stability subgroup of any element of the conjugacy class $\cyc(\mu)$,
 where $m_i(\mu)$ is the number of parts of $\mu$ equal to $i$.
Computing $H(\mu^{(1)}, \dots, \mu^{(k)})$ using this formula requires the character table for $S_N$, 
which becomes increasingly computationally complex \cite{Be} for rising $N$ and $k$.

To define {\em weighted} Hurwitz numbers, we introduce a weight generating function $G(z)$, either as an infinite sum
\be
G(z) = 1 + \sum_{i=1}^\infty g_i z^i
\label{G_z_taylor}
\ee
or an infinite product
\be
G(z) =\prod_{i=1}^\infty (1 + c_i z)
\label{G_z_prod}
\ee
or, in dual form
\be
\tilde{G}(z) = \prod_{i=1}^\infty (1-d_i z)^{-1},
\label{G_tilde_z_prod}
\ee
which may also be developed as an infinite sum, 
\be
\tilde{G}(z)= 1+ \sum_{i=1}^\infty \tilde{g}_i z^i,
\label{G_tilde_z_taylor}
\ee
either formally, or under suitable convergence conditions imposed upon the  parameters 
$\{c_i\}$, $\{d_i\}_{i\in \Nb^+}$ or the Taylor coefficients  $\{g_i\}_{i\in \Nb^+}$,   $\{\tilde{g_i}\}_{i\in \Nb^+}$.
The independent parameters determining the weighting may be viewed as  any of these.
Cases (\ref{G_z_prod}) and (\ref{G_tilde_z_prod}) may be viewed
as generating functions for elementary $\{e_i\}_{i\in \Nb}$ and complete symmetric functions 
$\{h_i\}_{i\in \Nb}$,, respectively, giving
\be
g_i = e_i({\bf c}), \quad \tilde{g}_i = h_i({\bf d}), \quad i\in \Nb^+,
\ee
where ${\bf c} =(c_1, c_2, \dots)$, ${\bf d} =(d_1, d_2, \dots)$.
  
The particular case of exponential weight generating function
\be
G(z) = \exp(z) = e^z
\ee
is of special interest, since it corresponds to a Dirac measure on the space of $k$-tuples of
partitions $\{\mu^{(1)}, \dots, \mu^{(k)}\}$, supported uniformly on $k$-tuples of $2$-cycles
\be
\mu^{(i)}  = (2, (1)^{N-2}),\quad i=1, \dots , k
\ee
which, from the viewpoint of enumeration of branched covers of the Riemann sphere, corresponds to
{\em simple} branching \cite{Ok, Pa}.
A number of other special cases  of particular interest will be considered
in  the following. 

The first of these is the family of rational weight generating functions 
\be
G_{{\bf c}, {\bf d}}(z) :={ \prod_{i=1}^L (1+c_i z) \over \prod_{j=1}^M(1-d_j z)},
\label{G_ratl_c_d}
\ee
for an arbitrary set of $L+M$ complex numbers $\{c_1, \dots, c_L, d_1, \dots , d_M\}$.
The Taylor series coefficients for this case are
\be
g_i({\bf c}, {\bf d}):= \sum_{j=0}^i e_j({\bf c}) h_{i-j}({\bf d}),
\label{rational_g_i}
\ee
where ${\bf c} =(c_1, \dots, c_L)$, ${\bf d}=(d_1, \dots, d_M)$.
Particular cases include:  linear or quadratic polynomials $(L,M) =(1,0)$ or $(2,0)$,
which correspond to two or three branch points (Belyi curves) \cite{Z, KZ} and the case $(L,M)=(0,1)$ which,
in the double weighted Hurwitz number case (which actually corresponds to a 2D Toda $\tau$-function),
is equivalent to the Harish-Chandra-Itzykson-Zuber matrix integral \cite{HC,IZ} as generating function
\cite{GGN1}. 

Also of particular interest is the case of the  $q$-exponential function
\be
G(z)=H_q(z):=\prod_{i=0}^\infty(1-q^i z)^{-1} =: e_q\left(z(1-q)^{-1} \right),
\label{Hq_def}
\ee
which is the weight generating function for (one  version of) {\em quantum Hurwitz numbers} \cite{GH2,H2},
whose Taylor coefficients are given by
\be
g_i(q):= {1\over (q\, ;q)_i}, \quad i \in \Nb^+,
\label{g_qqi}
\ee
where $(q\, ;q)_i$ is the $q$-Pochhammer symbol evaluated at $(q,q)$:
\be
(q\, ;q)_i := (1-q)\cdots (1-q^i).
\label{qqi_def}
\ee

\begin{definition}[Weighted Hurwitz numbers]
For the case (\ref{G_z_prod}), choosing a positive integer $d$ and a fixed partition $\mu$
of weight  $|\mu| =N$, the weighted (single) Hurwitz number $H^d_G(\mu)$ is defined as the weighted sum over all $k$-tuples $(\mu^{(1)}, \dots, \mu^{(k)})$
\be
H^d_G(\mu) := 
\sum_{k=1}^d  \sideset{}{'}\sum_{\mu^{(1)}, \dots \mu^{(k)}, \atop  |\mu^{(i)}| =N , \ \sum_{i=1}^k \ell^*(\mu^{(i)}) =d} 
\WW_G(\mu^{(1)}, \dots, \mu^{(k)}) H(\mu^{(1)}, \dots, \mu^{k)}, \mu),
\label{H_d_G_def}
\ee
where $\sideset{}{'}\sum$ denotes a sum over all  $k$-tuples of partitions 
$\{\mu^{(1)}, \dots, \mu^{(k)}\}$ of $N$ other than the cycle type of the identity element $(1^N)$,
\be
\ell^*(\mu^{(i)}) := |\mu^{(i)}| - \ell(\mu^{(i)})= N - \ell(\mu^{(i)})
\ee
is the {\em colength} of the partition $\mu^{(i)}$, and the weight factor is defined to be
\bea
 \WW_G(\mu^{(1)}, \dots, \mu^{(k)}) &\&:=
 {1\over k!}
 \sum_{\sigma \in S_{k}} 
 \sum_{1 \le b_1 < \cdots < b_{k } } 
  c_{b_{\sigma(1)}}^{\ell^*(\mu^{(1)})} \cdots c_{b_{\sigma(k)}}^{\ell^*(\mu^{(k)})} \cr
 &\&= {|\aut(\lambda)|\over k!} m_\lambda ({\bf c}).
 \label{WG_def}
\eea
Here  $m_\lambda ({\bf c}) $ is the monomial symmetric function of the parameters \hbox{${\bf c}:= (c_1, c_2, \dots)$}
\be
m_\lambda ({\bf c}) = {1\over |\aut(\lambda)|}\sum_{\sigma \in S_{k}} \sum_{1 \le b_1 < \cdots < b_{k }}
 c_{b_{\sigma(1)}}^{\lambda_1} \cdots c_{b_{\sigma(k)}}^{\lambda_{k}} ,
  \label{m_lambda}
\ee
 indexed by the partition $\lambda$ of weight $|\lambda|=d$ and length  $\ell(\lambda) =k$, whose 
parts $\{\lambda_i\}$  are equal to the colengths $\{\ell^*(\mu^{(i)})\}$ (expressed in weakly decreasing order) 
\be
\{\lambda_i\}_{i=1, \dots k} \sim \{\ell^*(\mu^{(i)})\}_{i=1, \dots k},  \quad \lambda_1 \ge \cdots \ge \lambda_k >0
\label{lambda_colengths_mu}
\ee
and
 \be
 |\aut(\lambda)|:=\prod_{i\geq 1} m_i(\lambda)!
 \label{def_aut_lambda}
 \ee
where $m_i(\lambda)$ is the number of parts of $\lambda$ equal to~$i$. 
We similarly denote by
\be
\tilde{H}^d_G(\mu) := 
\sum_{k=1}^d \sideset{}{'}\sum_{\mu^{(1)}, \dots \mu^{(k)}, \atop  |\mu^{(i)}| =N, \ \sum_{i=1}^k \ell^*(\mu^{(i)}) =d} 
\WW_G(\mu^{(1)}, \dots, \mu^{(k)}) \tilde{H}(\mu^{(1)}, \dots, \mu^{k)}, \mu)
\label{H_d_G_connected_def}
\ee
the {\em connected} weighted Hurwitz numbers corresponding to the weight generating function $G(z)$.
\end{definition}

Recall that the sum 
\be
d :=\sum_{i=1}^k\ell^*(\mu^{(i)})
\label{d_def}
\ee
of the colengths of the ramification profiles $\{\mu^{(i)}\}$ of the weighted branch points
determines the Euler characteristic $\chi$ of the covering curve through the Riemann-Hurwitz formula
\be
\chi =  N +\ell(\mu) -d,
\label{riemann_hurwitz}
\ee
which for connected coverings $\CC\ra \Pb$ is related, as usual, to the genus $g$ by
\be
\chi = 2 - 2g.
\label{chi_g}
\ee

\begin{remark}
Note that we can rearrange the sum \eqref{H_d_G_def} as follows;
\be
H^d_G(\mu) = \sum_{\lambda:\ |\lambda|=d} m_\lambda({\bf c}) \frac {|\aut(\lambda)|}{\ell(\lambda)!} \sum_{\{\ell_1,\dots, \ell_k\} \in Ana(\lambda)} \sum_{\ell^*(\mu^{(j))}=\ell_j:\atop |\mu^{(j)}|=|\mu|} H(\mu^{(1)}, \mu^{(2)},\dots, \mu^{(k)}, \mu)
\ee
where $k=\ell(\lambda)$. Here $Ana(\lambda) =\{(\ell_1, \dots, \ell_k)\}$ denotes the distinct rearrangements  
(or {\rm anagrams}) of the numbers $\lambda_1,\dots, \lambda_k$. But there are precisely 
$\frac {|\ell(\lambda)|!} {|\aut(\lambda)|}$ such anagrams, so if we denote 
\be
F_\lambda(\mu):= \sum_{\ell^*(\mu^{(j))}=\ell_j\atop |\mu^{(j)}|=|\mu|} H(\mu^{(1)}, \mu^{(2)},\dots, \mu^{(k)}, \mu),
\label{F_lambda_mu_def}
\ee
it is clear that $F_\lambda(\mu)$ is invarant under permutations of the $\ell_j$'s,.  The counting over anagrams 
therfore exactly cancels the factor  ${|\aut|\over k!}$, and (\ref{H_d_G_def}) is equivalent to
\be
H^d_G(\mu) = \sum_{\lambda, \ |\lambda|=d} m_\lambda({\bf c}) F_\lambda(\mu).
\ee
The same is true, of course, for the connected version $\wt H^d_G(\mu)$, with $F_\lambda(\mu)$ replaced by
\be
\tilde{F}_\lambda(\mu):= \sum_{\ell^*(\mu^{(j))}=\ell_j\atop |\mu^{(j)}|=|\mu|} \tilde{H}(\mu^{(1)}, \mu^{(2)},\dots, \mu^{(k)}, \mu).
\label{F_lambda_mu_connected_def}
\ee
\end{remark}

For the dual generating function $\tilde{G}(z)$, the weighted Hurwitz numbers are defined as
\be
H^d_{\tilde{G}}(\mu) := 
\sum_{k=1}^d \sum_{\mu^{(1)}, \dots \mu^{(k)}, \  |\mu^{(i)}| =N  \atop \sum_{i=1}^k \ell^*(\mu^{(i)}) =d} 
\WW_{\tilde{G}}(\mu^{(1)}, \dots, \mu^{(k)}) H(\mu^{(1)}, \dots, \mu^{k)}, \mu),
\label{H_d_G_tilde_def}
\ee
where
\bea
 \WW_{\tilde{G}}(\mu^{(1)}, \dots, \mu^{(k)}) &\&:=
 {(-1)^{\sum_{i=1}^k\ell^*(\mu^{(i)})-k}\over k!}
 \sum_{\sigma \in S_{k}} 
 \sum_{1 \le b_1\cdots \le b_k}
  d_{b_{\sigma(1)}}^{\ell^*(\mu^{(1)})} \cdots d_{b_{\sigma(k)}}^{\ell^*(\mu^{(k)})} \cr
 &\&= {|\aut(\lambda)|\over k!}  f_\lambda ({\bf d}).
 \label{WG_tilde_def}
\eea
Here  $f_\lambda ({\bf d}) $ is the ``forgotten''  symmetric function \cite{Mac} of the parameters 
\hbox{${\bf d}:= (d_1, d_2, \dots )$}
\be
f_\lambda ({\bf d}) = {(-1)^{\ell^*(\lambda)}\over |\aut(\lambda)|}\sum_{\sigma \in S_{k}} \sum_{1 \le b_1 \le \cdots \le b_{k }}
 d_{b_{\sigma(1)}}^{\lambda_1} \cdots d_{b_{\sigma(k)}}^{\lambda_{k}} ,
  \label{m_lambda}
\ee
 indexed again by the partition $\lambda$ of weight $|\lambda|=d$ and length  $\ell(\lambda) =k$, whose 
parts $\{\lambda_i\}$  are equal to the colengths $\{\ell^*(\mu^{(i)})\}$, as in (\ref{lambda_colengths_mu}).

For  rational weight generating functions $G_{{\bf c}, {\bf d}}(z)$ as defined  in (\ref{G_ratl_c_d}),
the weighted Hurwitz numbers are
\be
H^d_{G_{{\bf c}, {\bf d}}} := 
\sum_{k=1}^d \sum_{\mu^{(1)}, \dots \mu^{(k)}, \  |\mu^{(i)}| =N  \atop \sum_{i=1}^k \ell^*(\mu^{(i)}) =d} 
\WW_{G_{{\bf c}, {\bf d}}}(\mu^{(1)}, \dots, \mu^{(k)}) H(\mu^{(1)}, \dots, \mu^{k)}, \mu)
\label{H_d_G_rat_def}
\ee
where the weight factor is defined to be
\bea
 \WW_{G_{{\bf c},{\bf d}}}(\mu^{(1)}, \dots, \mu^{(k)}):=
{(-1)^{\sum_{i=1}^k\ell^*(\mu^{(i)})-k} \over k!}
&\& \sum_{\sigma , \sigma'\in S_{k}}   \sum_{1 \le a_1 < \cdots < a_{k } \le L\atop 1 \le b_1\cdots \le b_k\le M } 
  c_{a_{\sigma(1)}}^{\ell^*(\mu^{(1)})} \cdots c_{a_{\sigma(k)}}^{\ell^*(\mu^{(k)})}  \cr
  &\& \cr
&\&\times 
  d_{b_{\sigma'(1)}}^{\ell^*(\mu^{(1)})} \cdots d_{b_{\sigma'(k)}}^{\ell^*(\mu^{(k)})}.
 \label{WG_cd_def}
\eea

Finally, for the case of the generating function $H_q(z)$ defined  in ({\ref{Hq_def}}), the 
quantum weighted Hurwitz number is
\be
H^d_{H_q}(\mu) := 
\sum_{k=1}^d \sum_{\mu^{(1)}, \dots \mu^{(k)}, \  |\mu^{(i)}| =N  \atop \sum_{i=1}^k \ell^*(\mu^{(i)}) =d} 
\WW_{H_q}(\mu^{(1)}, \dots, \mu^{(k)}) H(\mu^{(1)}, \dots, \mu^{(k)}, \mu),
\label{H_d_Hq_def}
\ee
where the weight factor is 
\bea
 \WW_{H_q}(\mu^{(1)}, \dots, \mu^{(k)}):= \sum_{\sigma \in S_k} \prod_{j=1}^k
 {1  \over (1 - q^{\sum_{i=1}^j\ell^*(\mu^{(\sigma(i)})})}.
 \label{WHq_def}
\eea

  In order to compute the weighted Hurwitz numbers for any given weight generating function $G(z)$,
  it is necessary not only to know the pure Hurwitz numbers $H(\mu^{(1)}, \dots, \mu^{(k)}, \mu)$
  or  $\tilde{H}(\mu^{(1)}, \dots, \mu^{(k)}, \mu)$ entering in the sums (\ref{H_d_G_def}),
  (\ref{H_d_G_connected_def}), but  also to express the  weights ({\ref{WG_def}) or ({\ref{WG_tilde_def}) 
  as weighted homogeneneous polynomials in the Taylor coefficients $\{g_i\}_{i \in \Nb^+}$ 
  or $\{\tilde{g}_i\}_{i \in \Nb^+}$ of $G$ or $\tilde{G}$, respectively. This requires the transition matrices \cite{Mac}
  relating the various bases $\{m_\lambda\}$, $\{f_\lambda\}$ $\{e_\lambda\}$, $\{h_\lambda\}$ of the ring of symmetric functions  which, in turn,  require the Kostka matrices $K_{\lambda \mu}$ which are
   the transition matrices between the Schur function bases  $\{s_\lambda\}$  and  the monomial symmetric 
   function bases $\{m_\lambda\}$.  
   The latter is readily computable for any given weight $d$, but no known formula exists.  Therefore, the computation 
   of the weighted Hurwitz numbers $H^d_G(\mu)$ and their connected version  $\tilde{H}^d_G(\mu)$  directly from their definitions 
is a lengthy and complex task. As will be shown in the following, however, this may all be circumvented through
 the use of generating functions.
  
  \begin{remark}
Unlike the direct computation, based on the Frobenius-Schur formula (\ref{Frob_Schur_Hurwitz}),
the computation of weighted Hurwitz numbers $H^d_G(\mu)$ does not require any knowledge of the $S_N$
  group characters or the Kostka matrices,  and this is part of the reason for the more rapid computational algorithm 
  the generating function approach provides.
  
   However, it should be noted that knowing the weighted Hurwitz numbers $H^d_G(\mu)$, even for the most 
   general choice of weight generating functions $G(z)$, does not provide sufficient information to uniquely retrieve 
   the individual pure Hurwitz numbers $\{H(\mu^{(1)}, \dots, \mu^{(k)})\}$.
 This is because the only information retained in the definition of the weighted average (\ref{H_d_G_def})
is the partition $\lambda$ of length $k=\ell(\mu)$ and weight $|\lambda|=d$ whose parts are the 
colengths $\{\ell(\mu^{i)}\}_{i=1, \dots, k}$.

 Thus, knowing the weighted Hurwitz numbers allows us only to reconstruct the values of
 the pure Hurwitz numbers summed, as in (\ref{F_lambda_mu_def}), over all $k$-tuples of partitions $\{\mu^{(1)}, \dots, \mu^{(k)}\}$
 having the same sequence of colengths $\{\ell^*(\mu^{(1)}), \dots, \ell^*(\mu^{(k)})\}$. These 
 can be determined from the weighted ones only in the case  of {\em simple} Hurwitz numbers, 
since these are the only ones in the class where the colengths are $(1, 1, \dots, 1)$.
In general, there is no other combination of partitions for which the same holds true.
  \end{remark}


\subsection{Hypergeometric $\tau$-functions as generating functions of weighted Hurwitz numbers}

We recall the definition of  KP $\tau$-functions of hypergeometric type \cite{OrSc}, which serve as generating functions for weighted Hurwitz numbers \cite{GH2, HO, H1}.
For any weight generating function $G(z)$ of the above type, and nonzero parameter $\beta$,
we define two doubly infinite sequences of numbers $\{r_i^{(G, \beta)}, \rho_i\}_{i \in \Zb}$,
labeled by integers
\bea
r^{(G, \beta)}_i &\&:= \beta G( i\beta), \quad i \in \Zb,  \quad \rho_0 =1,
\label{r_G_beta_i_def} \\
\rho_i &\& := \prod_{j=1}^i r^{(G, \beta)}_j,\quad  \rho_{-i}:=\prod_{j=0}^{i-1}( r^{(G, \beta)}_{-j})^{-1}, 
\quad i\in \Nb^+,
\label{rho_i_def}
\eea
related by
\be
r_i^{(G, \beta)} = {\rho_i \over \rho_{i-1}},
\ee
where $\beta$ is viewed as a small parameter for which $G(i\beta)$ does not vanish for
any integer $i\in  \Zb$. (It is possible to extend this by requiring that, if $G(i\beta)=0$ for some smallest positive
integer $i\in \Nb^+$, then $\rho_j:=0$ for all $j\ge i$.

For the exponential generating function $G(z)=e^z$, we have
\be
r_i = \beta e^{i \beta}, \quad \rho_i = \beta^i e^{{1\over 2} i (i+1)\beta} ,\quad \rho_{-i} =  \beta^{-i} e^{{1\over 2} i(i-1)\beta}, \quad 
g_i = {1\over i!}.
\ee

For rational weight generating functions $G_{{\bf c}, {\bf d}}(z)$ as defined in (\ref{G_ratl_c_d}),
the parameters $\{\rho_i\}_{i\in \Zb}$, and coefficients $\{g_i\}_{i\in \Nb}$ become
\bea
\rho_i := \beta^i\prod_{k=1}^i {\prod_{i=1}^L (1+k c_i \beta) \over \prod_{j=1}^M (1-kd_j\beta) }, &\&
\quad  \rho_{-i} := \beta^{-i}\prod_{k=1}^{i-1} {\prod_{j=1}^M (1+kd_j \beta) \over \prod_{i=1}^L(1-k c_i \beta) }, 
\quad i\in \Zb ,
\label{rho_c_b_i_def}
\cr
&&\\
g_i &\&= \sum_{j=0}^i e_j({\bf c}) h_{i-j}({\bf d}).
\label{gi_c_b_def}
\eea
For the quantum exponential generating function $H_q(z)$  defined in (\ref{Hq_def}),
the parameters $\{\rho_i\}_{i\in \Zb}$ and coefficients $\{g_i\}_{i\in \Nb}$ become
\bea
\rho_i := \beta^i\prod_{k=1}^i \prod_{j=0}^\infty  (1- k q^j \beta)^{-1}&\&
\quad  \rho_{-i} := \beta^i\prod_{k=1}^{i-1} \prod_{j=0}^\infty  (1+ k q^j \beta),
\quad i\in \Zb ,
\label{rho_Hq_def}
\cr
&\& \\
g_i &\&= {1\over (q;q)_i}.
\label{gi_Hq_def}
\eea

For each partition $\lambda$ of $N$, we define the associated {\em content product} coefficient
\be
r^{(G, \beta)}_\lambda := \prod_{(i,j) \in \lambda} r^{(G, \beta)}_{j-i}.
\ee
The KP $\tau$-function of hypergeometric type associated to these parameters is defined as the Schur function series \cite{HO,GH1, GH2}:
\be
\tau^{(G, \beta)}({\bf t}) 
:=\sum_{N=0}^\infty \sum_{{\lambda}, \ |\lambda|=N}  
(h(\lambda))^{-1} r_\lambda^{(G, \beta)}  s_\lambda({\bf t}),
\label{tau_G_beta_schur_series}
\ee
where $h(\lambda)$ is the product of the hook lengths of the partition $\lambda$
 and \hbox{${\bf t}=(t_1, t_2 \dots)$} is the infinite sequence of KP flow parameters, which may be equated to the 
 sequence  of normalized power sums $(p_1, {1\over 2} p_2, \dots )$
\be
t_i := {1\over i} \sum_{a} x_a^i =  {1\over i} p_i
\ee
 in a finite or infinite set of auxiliary variables $(x_1, x_2, \dots )$.

Using the Schur character formula \cite{FH, Mac}
  \be
  s_\lambda = \sum_{\mu, \ |\mu|=|\lambda| } {\chi_\lambda(\mu) \over z_\mu} p_\mu,
  \label{schur_character_formula}
  \ee
  where $\chi_\lambda(\mu)$ is the irreducible character  determined by $\lambda$ evaluated on the conjugacy class 
  $\cyc(\mu)$, $z_{\mu}$  is the order of the stabilizer of the elements of this conjugacy class,  as given in (\ref{z_mu_def}), and
  \be
  p_\mu := \prod_{i=1}^{\ell(\mu)} p_{\mu_i}
  \label{power_sum_mu}
  \ee
  is the power sum symmetric function corresponding to partition $\mu$, we may re-express the Schur function
  series (\ref{tau_G_beta_schur_series}) as an expansion  in the basis $\{p_\mu\}$ of power sum symmetric functions.
As  shown in \cite{GH1, H1, HO, GH2}, this gives the following:
  \begin{theorem}
  \label{tau_gener_fn_weighted_hurwitz}
  The $\tau$-function $\tau^{(G, \beta)}({\bf t}) $ may equivalently be expressed as
 \be
\tau^{(G, \beta)}({\bf t}) 
=\sum_{\mu}\sum_{d=0}^\infty  \beta^d H^d_G (\mu) p_\mu({\bf t}).
 \label{tau_G_beta_power_sum_series}
 \ee
It is thus a generating function for the weighted Hurwitz numbers $H^d_G (\mu) $.
This holds similarly for the dual weight generating function $\tilde{G}(z)$ and the weighted
Hurwitz numbers $H^d_{\tilde{G}}(\mu)$
\be
\tau^{(\tilde{G}, \beta)}({\bf t}) 
=\sum_{\mu} \sum_{d=0}^\infty  \beta^d H^d_{\tilde{G}} (\mu) p_\mu({\bf t}).
 \label{tau_G_tilde_beta_power_sum_series}
 \ee
     \end{theorem}
    
    Particular cases of weighted Hurwitz numbers for which hypergeometric  $\tau$-functions correspond to the weight generating functions $G(z)$ defined above include:  {\em simple} Hurwitz numbers  \cite{Ok, Pa} (both single and double), with weight generating function $G(z)=e^z$; {\em weakly monotone} Hurwitz numbers \cite{GGN1} (or, equivalently, {\em signed} Hurwitz numbers \cite{HO, GH2}),  with weight generating function $G(z) = {1\over 1-z}$;   {\em strongly monotone} Hurwitz numbers \cite{GH1, GH2}, with weight generating function $G(z) =1 + z$ (or,  equivalently, weighted Hurwitz numbers for Belyi curves and {\em dessins d'enfant} \cite{AC1, KZ, Z});  polynomially weighted Hurwitz numbers \cite{ AMMN, HO, ACEH1, ACEH2, ACEH3}; 
{\em quantum} Hurwitz numbers \cite{H1, GH2, H2} and  {\em multispecies} Hurwitz numbers \cite{H3}.

\subsection{The pair correlator $K^G(x,y)$}

An essential r\^ole in the computation of weighted Hurwitz numbers 
 is  played by the {\em pair correlator} associated to the $\tau$-function, defined as
\be
K^G(x,y) := {\tau^{(G, \beta)}([x] -[y])\over x-y},
\ee
where $[x]$ denotes the infinite sequence
\be
[x] = (x, x^2/2, x^3/3, \dots, x^n/n, \dots).
\label{rect_x_def}
\ee
This has the following series expansion \cite{ACEH1, ACEH3}
\be
K^G(x,y)  ={1 \over x - y} +  \beta^{-1} \sum_{a=0}^\infty \sum_{b=0}^\infty (-1)^b
{\rho_a  \rho^{-1}_{-b-1} \over a! b!(1+a+b)} \left({x\over \beta}\right)^a \left({y\over \beta}\right)^b.
\label{Khook} 
\ee
Defining
\be
\rho_{ab}(\beta) := (-1)^{b}\prod_{i=-b}^a {G(i \beta) \over a! b! (a+b+1)},
\label{rho_ab}
\ee
we have
\be
\rho_{ab}(\beta) = (-1)^{a+b}\rho_{ba}(-\beta)
\label{rho_ab_sym}
\ee
and the nonsingular part of $K^G(x,y)$ is
\be
K_0^G(x,y)  := K^G(x,y) - {1\over x-y} = \sum_{a=0}^\infty \sum_{b=0}^\infty \rho_{ab} x^a y^b.
\label{K0_exp}
\ee

Taking the (formal) Taylor series expansion
\be
\rho_{ab}(\beta) := \sum_{d=0}^\infty \rho_{ab}^d \beta^d,
\label{rho_ab_series}
\ee
the coefficents $\rho_{ab}^d({\bf g})$ are weighted degree $d$ homogeneous polynomials in
the Taylor coefficents $\{g_i\}_{i\in \Nb^+}:={\bf g} $ of  $G(z)$.
It follows from (\ref{rho_ab_sym}) that they satisfy the symmetry conditions
\be
\rho_{ab}^d = (-)^{a+b+d} \rho^d_{ba}.
\label{rho_abd_sym}
\ee
Their values for  small values of $a,b, d$ are displayed in Tables \ref{rho_table1}- \ref{rho_table3} of Appendix \ref{rho_abc_coeffs}.

In particular, for the case of the rational weight generating function $G_{{\bf c}, {\bf d}}(z)$,
we have
\be
K^{G_{{\bf c}, {\bf d}}}(x,y)  ={1 \over x - y} + \sum_{a=0}^\infty \sum_{b=0}^\infty  (-1)^b
\left(\prod_{k=-b}^{a}{\prod_{i=1}^M (1 + ikc_i\beta)  
\over \prod_{j=1}^M (1 - k d_j\beta) }\right)
 {x^a y^b \over  a!b!(1+a+b)},
\label{Khook_c_d} 
\ee
 and for the quantum  exponential weight generating function $G_{H_q}(z)$,
\be
K^{H_q}(x,y)  ={1 \over x - y} + \sum_{a=0}^\infty \sum_{b=0}^\infty  (-1)^b
\left(\prod_{k=-b}^{a} \prod_{j=0}^\infty  (1- kq^j \beta)^{-1} \right)
 {x^a y^b \over  a!b!(1+a+b)}.
\label{Khook_c_d} 
\ee

\subsection{The multicurrent correlator $W_n(x_1, \dots, x_n)$ as generating function for  weighted Hurwitz numbers}
\label{multicurrent_correl}

Following \cite{ACEH1, ACEH2, ACEH3}, define the derivations:
\begin{definition}
For any parameter $x$
\be
\nabla(x) := \sum_{i=1}^\infty x^{i-1}{\partial \over \partial t_i}, \quad \tilde{\nabla}(x) := \sum_{i=1}^\infty{ x^{i}\over i}{\partial \over \partial t_i} 
\ee
\end{definition}
In terms of these, we introduce the following correlators
\bea
W^G_{n}( x_1,\dots,x_n )  := \left.\left( \Big(\prod_{i=1}^n \nabla(x_i) \Big)
\tau^{(G, \beta)}({\bf t})\right) \right\vert_{{\bf t}= {\bf 0}},
 \label{W_G_def}\\
 \tilde{W}^G_{n}( x_1,\dots,x_n )  := \left.\left( \Big(\prod_{i=1}^n \nabla(x_i) \Big)
\ln\,  \tau^{(G, \beta)}({\bf t}) \right) \right\vert_{{\bf t}= {\bf 0}},
\label{W_tilde_G_def} \\
F^G_{n}( x_1,\dots,x_n )  := \left.\left( \Big(\prod_{i=1}^n \tilde{\nabla}(x_i) \Big)
\tau^{(G, \beta)}({\bf t})\right) \right\vert_{{\bf t}= {\bf 0}},
 \label{F_G_def}\\
 \tilde{F}^G_{n}( x_1,\dots,x_n )  := \left.\left( \Big(\prod_{i=1}^n \tilde{\nabla}(x_i) \Big)
\ln\, \tau^{(G, \beta)}({\bf t}) \right) \right\vert_{{\bf t}= {\bf 0}},
\label{F_tilde_G_def}
\eea
which are related by
\bea
\label{eq:defWGn}
W^G_{n}(x_1,\dots,x_n) &\&= \frac{\partial}{\partial x_1}\dots \frac{\partial }{\partial x_n} F^G_{n}(x_1,\dots,x_n),\label{WGn_dFGn}
\\
\tilde{W}^G_{n}(x_1,\dots,x_n) &\&=\frac{\partial}{\partial x_1}\dots \frac{\partial}{ \partial x_n}\tilde{F}^G_{n}(x_1,\dots,x_n).
\label{tilde_WGn_dFGn}
\eea
As shown in \cite{ACEH2},  $W_n(x_1, \dots, x_n)$ has a fermionic representation  as a multicurrent correlator.
Moreover, we have the following result from refs.~\cite{ACEH1, ACEH3},  which shows that $F_n(x_1,\dots,x_n)$  and $\tilde{F}_n(x_1,\dots,x_n)$
may also serve as generating functions for weighted single Hurwitz numbers $H^d_G(\mu)$ and
 $\tilde{H}^{d}_G(\mu)$, respectively,  with  ramification profile  $\mu$  at the single marked branch 
 point of length $\ell(\mu)=n$. 
\begin{proposition}[\cite{ACEH1, ACEH3}]
\label{F_n_gen_fn}
\bea
F^G_n(x_1,\dots,x_n) &\&=  \sum_{d=0}^\infty \sum_{\mu \atop \ell(\mu)=n} \beta^d H_{G}^d(\mu)\, |\aut(\mu)| m_\mu(x_1,\dots,x_n),
\label{Fn_expansion}
\\ 
\tilde{F}^G_n(x_1,\dots,x_n) &\&:= \sum_{d=0}^\infty \sum_{\mu \atop \ell(\mu)=n}  \beta^d \tilde H^d_{G}(\mu)\,|\aut(\mu)|m_\mu(x_1,\dots,x_n) 
\label{tilde_Fn_expansion}
\\
&\&=  \sum_{g=0}^\infty  \beta^{2g-2+n} \tilde{F}^G_{g,n}(x_1, \dots, x_n), 
\eea
where
\be
\label{tilde_Fgn_expansion}
\tilde{F}^G_{g,n}(x_1,\dots,x_n) =  \sum_{\mu\atop \ell(\mu)=n}  \tilde H^{2g-2+n+|\mu|}_{G}(\mu)\, |\aut(\mu)| m_\mu(x_1,\dots,x_n)
\ee
and $m_\mu(x_1,\dots, x_n)$ is the monomial symmetric polynomial in the indeterminates $(x_1,\dots , x_n)$.
\end{proposition}

Finally, we quote the following further result from \cite{ACEH1, ACEH3}, which expresses the connected multicurrent
correlators in terms of pair correlators.

\begin{proposition}
\label{prop:detConnected}
\be
\label{eq:detConnected1}
\tilde W^G_1(x) = \lim_{x'\to x} \left( K^G(x,x') - \frac{1}{x-x'}\right),
\ee
\be
\label{eq:detConnected2}
\tilde W^G_2(x_1,x_2) = \left( -K^G(x_1,x_2)K^G(x_2,x_1)  - \frac{1}{(x_1-x_2)^2}\right),
\ee
and for $n\ge 3$
\be\label{eq:detConnectedn}
\tilde W^G_n(x_1,\dots,x_n) =
\sum_{\sigma \in S_n^{\rm 1-cycle}}
\sgn(\sigma) \prod_{i} K^G(x_i,x_{\sigma(i)}),
\ee
where the last sum is over all permutations in $S_n$  consisting of a single $n$-cycle.
\end{proposition}

The corresponding nonconnected quantities are obtained  from the cumulant relations \cite{ACEH3}:
  \bea
W_1(x_1)  = &\&\tilde{W}_1(x_1 ) \cr
W_2(x_1, x_2) =  &\&\tilde{W}_2(x_1, x_2)  + \tilde{W}_1(x_1) \tilde{W}_1(x_2) \cr
  W_3(x_1, x_2, x_3)  = &\& \tilde{W}_3(x_1, x_2, x_3 ) 
   + \tilde{W}_1(x_1) \tilde{W}_2(x_2, x_3) \cr
   &\&+ \tilde{W}_1(x_2) \tilde{W}_2(x_1, x_3)   + \tilde{W}_1(x_3) \tilde{W}_2(x_1 x_2) \cr
   &\& + \tilde{W}_1(x_1) \tilde{W}_1(x_2) \tilde{W}_1(x_3)
  \label{tilde_W123_cumulant}
 \eea
or, more generally,or $n\ge 3$,
\bea
W_n(x_1,\dots,x_n)&=&
\sum_{\ell\geq 1} \sum_{I_1\uplus\dots\uplus I_\ell = \{1,\dots,n\}}
\prod_{i=1}^\ell \tilde{W}_{|I_i|}(x_j,j\in I_i),
\label{cumulant_Wn}
\eea
 with identical relations holding between the $\tilde{F}_n$'s and $F_n$'s .


\section{Generating function computation of $\tilde{H}^d_G(\mu)$}
\label{W_123_gener}

Using the series expansion (\ref{Khook_c_d}) for the pair correlator $K^G(x,y)$, 
and expressing the connected multicurrent correlators $\tilde{W}_n(x_1, \dots ,x_n)$ in terms of these via eqs.~(\ref{eq:detConnected1}), ({\ref{eq:detConnected2}), ({\ref{eq:detConnectedn}), we can compute all
weighted connected or nonconnected single Hurwitz numbers $\tilde{H}^d_G(\mu)$, $H^d_G(\mu)$ for coverings of any degree $N$, Euler characteristic  $\chi$ 
and unweighted partition $\mu$ of $N$,
as  homogeneous  polynomials of weighted  degree $d$ in the coefficients
$\{g_i\}_{i\in \Nb^+}$ of the series expansion of the weight generating function $G(z)$.

In particular, for $\ell(\mu) \le 3$, we have:
\begin{theorem}
\label{weighted_hurwitz_formulae}
[Weighted Hurwitz numbers for $\ell(\mu)=1,2$ or $3$]

The following formulae determine the connected weighted Hurwitz numbers
$\tilde{H}^d_G(\mu)$ for  $\mu$ of length $\ell(\mu) \le 3$.

$\ell((\mu_1)) =1$:
\be
\tilde{H}^d_G((\mu_1)) = H^d_G((\mu_1))= {1\over \mu_1}\sum_{a=0}^{\mu_1-1}\rho^d_{a, \mu_1 - a-1}.
\label{Hd1_G}
\ee

$\ell((\mu_1. \mu_2)) =2$:
\be
\tilde{H}^d_G((\mu_1, \mu_2))  = \tilde{H}^{(d,2)}_{(G,1)}((\mu_1, \mu_2))  +  \tilde{H}^{(d,2)}_{(G,2)}((\mu_1, \mu_2)),
\ee
where the linear part $\tilde{H}^{(d,2)}_{(G,1)}((\mu_1, \mu_2))$ is given by
\be
\tilde{H}^{(d,2)}_{(G,1)}((\mu_1, \mu_2)) = \small{{1 + (-1)^{d +|\mu |}\over \mu_1 \mu_2 |\aut(\mu_1,\mu_2)|}}
\sum_{b=0}^{\mu_2 -1}\rho^d_{b,\mu_1+ \mu_2- b- 1}
\label{Hd2_G1}
\ee
and the quadratic part $\tilde{H}^{(d,2)}_{G,2}((\mu_1, \mu_2))$ by
\be
\tilde{H}^{(d,2)}_{(G,2)}((\mu_1, \mu_2)) =- \small{{1\over \mu_1 \mu_2 |\aut(\mu_1,\mu_2)|}}
\sum_{a=0}^{\mu_1 -1}\sum_{b=0}^{\mu_2 -1}\sum_{j=0}^d\rho_{a, \mu_2-b-1}^j \rho^{d-j}_{b, \mu_1-a-1}.
\label{Hd2_G2}
\ee

$\ell((\mu_1, \mu_2, \mu_3)) =3$:
\be
\tilde{H}^{(d,3)}_G((\mu_1, \mu_2, \mu_3))  = 
\tilde{H}^{(d,3)}_{(G,1)}((\mu_1, \mu_2,, \mu_3))  +  \tilde{H}^{(d,3)}_{(G,2)}((\mu_1, \mu_2, \mu_3))
+ \tilde{H}^{(d,3)}_{(G,3)}((\mu_1, \mu_2,, \mu_3)),
\label{Hd3_sum}
\ee
where the linear part $\tilde{H}^{(d,3)}_{(G,1))}((\mu_1, \mu_2, \mu_3)) $ is given by
\be
\tilde{H}^{d}_{(G,1)}((\mu_1, \mu_2, \mu_3))=\small{{1 - (-1)^{d +|\mu |}\over \mu_1 \mu_2 \mu_3 |\aut(\mu_1,\mu_2, \mu_3)|}} \sum_{b=0}^{\mu_2 -1} \left(\rho^d_{\mu_2-b-1, \mu_1 + \mu_3 +b} 
-\rho^d_{\mu_1 + \mu_2 -b-1, \mu_3 +b}\right),
\label{Hd3_G1}
\ee
the quadratic part by 
\be
\tilde{H}^{(d,3)}_{(G,2)}((\mu_1, \mu_2, \mu_3)) = K^d_{(G,3)}(\mu_1, \mu_2; \mu_3) + K^d_{(G,3)}(\mu_1, \mu_3; \mu_2)+ K^d_{(G,3)}(\mu_3, \mu_2; \mu_1),
\label{Hd3_G2}
\ee
where
\bea
K^d_{(G,3)}(\mu_1, \mu_2; \mu_3) := &\&\small{{ (1 - (-1)^{d+|\mu|)} \over \mu_1 \mu_2 \mu_3 |\aut(\mu_1,\mu_2, \mu_3)|}} 
{\hskip -20pt} \sum_{a=0}^ {\min(\mu_1-1, {1 \over 2} (\mu_1 + \mu_2) -1)}    \sum_{c=0}^{\mu_3 -1}   \sum_{j=0}^d 
 \rho^j_{ c, \mu_1 + \mu_2 -a-1} \rho^{d-j}_{a,\mu_3 -c -1},\cr
  &\&
\eea
with $\aut(\mu_1,\mu_2, \mu_3)$ for a {\em composition}  $(\mu_1,\mu_2, \mu_3)$ (i.e., an unordered set of 
of positive integers)  defined in the same way (\ref{def_aut_lambda}) as for a partition,
and the cubic part by
\bea
\tilde{H}^{(d,3)}_{(G,3)}((\mu_1, \mu_2, \mu_3)) &\&= \small{{(1 - (-1)^{d+|\mu|})\over \mu_1 \mu_2 \mu_3 |\aut(\mu_1,\mu_2, \mu_3)|}}
\sum_{a=0}^{\mu_1 -1}\sum_{b=0}^{\mu_2 -1}\sum_{c=0}^{\mu_3 -1} \sum_{k=0}^d\sum_{j=0}^{d-k}\cr
&\&\times\rho^j_{a, \mu_2-b-1} \rho^k_{b, \mu_3 -c -1}\rho^{d-j-k}_{c, \mu_1-a-1}. \cr
&\&
\label{Hd3_G3}
\eea
\end{theorem}
\begin{proof}
Compare the expansion (\ref{tilde_Fn_expansion}) with  the expressions in
Proposition  \ref{prop:detConnected} for $\tilde{W}_1$, $\tilde{W}_2$, $\tilde{W}_3$ as polynomials in the pair correlators.
Substitute the series expansion (\ref{K0_exp}) for the $K(x_i,x_j)$'s  
and  the series expansion for  $\rho_{ab}^d$ in $\beta$ following from eq.~(\ref{rho_ab}), 
and use relation (\ref{tilde_WGn_dFGn}).
Equating like powers of $\beta$ gives formulae (\ref{Hd1_G}) - (\ref{Hd3_G3}) for the weighted connected Hurwitz numbers 
$\tilde{H}^d_G(\mu)$  as polynomials in the Taylor coefficients $\{g_i\}$  of $G(z)$.
\end{proof}

\begin{remark}
The nonconnected versions of the weighted Hurwitz  numbers are obtained by applying
the cumulant relations (\ref{tilde_W123_cumulant}), which give
\bea
H^d_G((\mu_1))&\& = \tilde{H}^d_G((\mu_1)), \\
H^d_G((\mu_1,\mu_2))&\& = \tilde{H}^d_G((\mu_1,\mu_2))+ {1\over |\aut(\mu_1, \mu_2))|}
\sum_{k=0}^d H^k_G((\mu_1))H^{d-k}_G((\mu_2)), \\
H^d_G((\mu_1, \mu_2, \mu_3))
&\&=\tilde{H}^d_G((\mu_1, \mu_2, \mu_3))
+{1\over |\aut(\mu_1,\mu_2, \mu_3)|}
\sum_{k=0}^d  \Big{(} |\aut(\mu_2, \mu_3)| H^k_G((\mu_1)) H^{d-k}_G((\mu_2, \mu_3))\cr
&\&
+ 
  |\aut(\mu_1, \mu_3)| H^k_G((\mu_2) )H^{d-k}_G((\mu_1, \mu_3))+ 
   |\aut(\mu_1, \mu_2)| H^k_G((\mu_3)) H^{d-k}_G((\mu_1, \mu_2)) \cr
  &\&  -2  \sum_{j=0}^{d-k}H^j((\mu_1))  H^k((\mu_2))  H^{d-j-k}((\mu_3)) \Big{)}.
\eea
\end{remark}
\begin{remark}
Similar formulae, consisting of degree $n$ polynomials in the coefficients $\rho^d_{ab}({\bf g})$,
may be obtained for all values of $n=\ell(\mu)$, determining thereby  the corresponding weighted Hurwitz
numbers $\tilde{H}^d_G(\mu)$, $H^d_G(\mu)$ from the series expansions of $\tilde{W}_n(x_1, \dots , x_n)$
for all $n$.
\end{remark}

\begin{remark}
Inserting the expressions for the coefficients $\rho_{ab}^d({\bf g})$, following from (\ref{rho_ab}), (\ref{rho_ab_series}), as graded homogeneous 
polynomials in the Taylor coefficients $\{g_i\}_{i\in \Nb}$ of the generating function $G(z)$, gives $\tilde{H}^d_G(\mu)$
as a weight $d$ graded homogeneous polynomial in the $g_i$'s.
\end{remark}

Tables \ref{Hurwitz_table4} - \ref{Hurwitz_table7}  of Appendix \ref{eval_weighted_hurwitz} display
$H^d_G(\mu)$ and $\tilde{H}^d_G(\mu)$ for generic values of the Taylor coefficients  $\{g_i\}_{i\in \Nb}$,
 $n  = 1,2,3$, and small values of $N$ and $d$. Their evaluation for the special cases of: exponential, 
 rational and quantum exponential weight generating functions  $e^z$,  $G_{({\bf c,} {\bf d})}(z)$ and $H_q(z)$, 
 respectively, are displayed in Tables \ref{Hurwitz_table8} - \ref{Hurwitz_table13}.

\renewcommand{\labelenumi}{\Alph{enumi})}
  
   \appendix
\section*{Appendices}

\section{Table of coefficents $\rho_{ab}^d({\bf g})$}
\label{rho_abc_coeffs}

\begin{table}[H]
\label{rho_table1}
\small{
\caption{$\rho_{ab}^1({\bf g})$ for $0 \le a, b, \le 4$}}
\center{
\begin{tabular}{c c c c c c c}\\
{\hskip 12 pt}$a\big{\backslash} b$  & $\big{|}$ & $0$  & $1$ & $2$ & $3$ &$4$ 
\\
\hline
\hline

$0 $ & $\bigg{|}$ & $0$ & ${1\over 2}g_1$ & $-{1\over 2}g_1$ &${1\over 4} g_1$ &  $-{1\over 12} g_1$ 
 \\
\hline
1 &\big{|}  & ${1\over 2}g_1$ & $0$& $-{1\over 4} g_1$ & ${1\over 6}g_1$ & $-{1\over 16}g_1$
\\
\hline
2 &\big{|}   &${1\over 2}g_1$ & $-{1\over 4}g_1$& $0$& ${1\over 24}g_1$ & $-{1\over 48}g_1$
\\
\hline
3 &\big{|}   &${1\over 4}g_1$ & $-{1\over 6}g_1$& ${1\over 24}g_1$& $0$ & $-{1\over 288}g_1$
\\
\hline 
4  &\big{|}  & ${1\over 12}g_1$ & $-{1\over 16}g_1$ & ${1\over 48}g_1$ & $-{1\over 288}g_1$ & $0$
\\
\hline
\hline
\end{tabular}}
\end{table}

\begin{table}[H]
\label{rho_table2}
\small{
\caption{$\rho_{ab}^2({\bf g})$ for $0 \le a,b \le 4$. The last column $\rho^2_{a4}({\bf g})$ is determined by
applying the symmetry property (\ref{rho_abd_sym}) and setting $\rho^2_{4,4} = -{5\over 864} (g_1^2 - 2 g_2)$.}}
\center{
\begin{tabular}{c c c c c c }\\
{\hskip 12 pt} $a\big{\backslash} b$  & $\big{|}$ & $0$  & $1$ &  $2$ & $3$ 
\\
\hline
\hline

$0 $ & $\big{|}$ & $0$  & $-{1\over 2}g_2$ & ${1\over 6} (2 g_1^2 + 5 g_2)$ & $-{1\over 24} (11 g_1^2 + 14 g_2)$ 
\\
\hline
$1$ &\big{|} & ${1\over 2}g_2$ & ${1\over 3} (g_1^2 - 2 g_2)$ & ${1\over 8} (-g_1^2 + 6 g_2)$ & $-{1\over 6 }(g_1^2 + 3 g_2)$ 
\\
\hline
$2$ &\big{|}  &${1\over 6} (2 g_1^2 + 5 g_2)$ & ${1\over 8} (g_1^2 - 6 g_2)$ & ${1\over 4} (-g_1^2 + 2g_2)$ & ${1\over 72}(5 g_1^2 - 19 g_2)$ 
\\
\hline
$3$ &\big{|}   & ${1 \over 24} (11 g_1^2 + 14 g_2)$ & $-{1\over 6} (g_1^2 + 3 g_2)$ & ${1\over 72} (-5 g_1^2 + 19 g_2)$ & ${1\over18}( g_1^2 - 2g_2)$ \\ 
\hline
$4$  &\big{|}  & ${1\over 24} (7g_1^2 + 6 g_2)$ & $-{1\over 144} (25 g_1^2 + 31 g_2)$ & ${1\over 48} (g_1^2 + 5 g_2)$
& ${1\over 576} (7 g_1^2 - 22 g_2)$ 
\\
\hline
\hline
\end{tabular}}
\end{table}

\begin{table}[H]
\label{rho_table3}
\small{
\caption{$\rho_{ab}^3({\bf g})$ for $0 \le a, b\le 4$. The last column $\rho^2_{a4}({\bf g})$ is determined
by applying the symmetry property (\ref{rho_abd_sym}) and setting  $\rho^3_{4,4} = 0 $.}}
\center{
\begin{tabular}{c c c c c c }\\
{\hskip 12 pt} $a\big{\backslash} b$ & $\big{|}$ & $0$  & $1$ & $2$ & $3$  
\\
\hline
\hline
$0 $ & $\big{|}$ &$0$ & ${1\over 2}g_3$ & $-{1\over 6} (6 g_1 g_2 - 9 g_3)$ & ${1\over 4} ({ g_1^3 + 8 g_1 g_2 \atop + 6 g_3})$ 
\\
\hline
$1$ &\big{|} &${1\over 2}g_3$ & $0$ & ${1\over 4} ({g_1^3 - 2 g_2g_2 \atop - 4 g_3})$ & ${1\over 6} ({-g_1^3 + 8 g_1g_2 \atop + 7 g_3})$ 
 \\
\hline
$2$ &\big{|} &${1\over 6} (6 g_1 g_2+ 9 g_3)$ & ${1\over 4} ({g_1^3 - 2 g_1 g_2 \atop - 4 g_3})$ & $0$ & ${1\over 24} ({-5 g_1^3 + 10 g_1 g_2 \atop + 9 g_3})$ 
\\
\hline
$3$ &\big{|} &${1\over 4} ({ g_1^3 + 8 g_1 g_2\atop+ 6 g_3})$ & ${1\over 6} ({g_1^3 - 8 g_1 g_2\atop  - 7 g_3})$ 
& ${1\over 24} ({-5 g_1^3 + 10 g_1g_2\atop + 9 g_3})$ & $0$ 
\\
\hline 
$4$ &\big{|} &${5\over 12} ({g_1^3 + 4 g_1 g_2 \atop + 2 g_3})$ & $-{1\over 48 }({5 g_1^3 + 60 g_1 g_2 \atop + 33 g_3})$ 
& ${1\over 48} ({-5 g_1^3 + 22 g_1 g_2 \atop + 13 g_3})$ & ${1\over 144} ({7 g_1^3 - 14 g_1 g_2\atop  - 8 g_3})$ 
\\
\hline
\hline
\end{tabular}}
\end{table}


\section{Evaluation of  weighted Hurwitz numbers}
\label{eval_weighted_hurwitz}

\subsection{Generic case}
\label{generic_case}

\begin{table}[H]
\label{Hurwitz_table4}
\small{
\caption{\small{Nonconnected weighted Hurwitz numbers $H^d_G(\mu)$ with $\ell(\mu)=1$ or $2$, \ $N=2,3$ or $4$}}
\begin{tabular}{c c c c }\\
$N$ & $\mu$ & $H_G^{N+\ell(\mu)-2}$ & $H_G^{N+\ell(\mu)}$  \\
\hline
\hline
2&(2) & ${1\over 2} g_1$ & ${1\over 2} g_3 $\\
\hline
3&(3)& ${1\over 3}(g_1^2 + g_2) $
& ${1\over 3} (g_2^2 + 4 g_1 g_3 + 5 g_4)$ \\
\hline
3&(21)& $g_1 g_2 + {3\over 2}g_3 $
&$ 3 g_1 g_4 + 2g_2 g_3 + {11\over 2} g_5 $\\
\hline
4 & (4) & ${1\over 4} (g_1^3 + 3 g_1 g_2 + g_3)$ 
&${5\over 2}g_1^2 g_3 + {5\over 4}g_1g_2^2 + {25\over 4} g_4 g_1 + {15\over 4} g_2 g_3 + {15\over 4}g_5$ \\
\hline
4 & (31) & $g_1^2 g_2  +  {4\over 3}g_2^2 + {10\over 3} g_1 g_3 +{8\over 3}g_4 $
& $ g_1(6 g_1 g_4 +8 g_2 g_3 +24 g_5 ) + g_2(g_2^2 +16 g_4)+7 g_3^2 +22 g_6$\\ 
\hline
4 & (22) &${1\over 2}g_1^2 g_2 +{1\over 4}g_2^2 +  g_1 g_3   +  {1\over 2}g_4$
& $ g_1( {11\over 4} g_1 g_4 +{13\over 4} g_2  g_3 + 9 g_5) 
+g_2( {1\over 4}g_2^2 +{19\over 4} g_4 )  +{11\over 4}g_3^2 + {25\over 4}g_6$\\ 
\hline
\end{tabular}}
\end{table}

\begin{table}[H]
\label{Hurwitz_table5}
\small{
\caption{\small{Connected weighted Hurwitz numbers $\tilde{H}^d_G(\mu)$ with $\ell(\mu)=1$ or $2$, \ $N=2,3$ or $4$}}
\begin{tabular}{c c c c }\\
$N$ & $\mu$ & $\tilde{H}_G^{N+\ell(\mu)-2}$ & $\tilde{H}_G^{N+\ell(\mu)}$  \\
\hline
\hline
2&(2) & ${1\over 2} g_1$ & ${1\over 2} g_3 $\\
\hline
3&(3)& ${1\over 3}(g_1^2 + g_2) $
& ${1\over 3} (g_2^2 + 4 g_1 g_3 + 5 g_4)$ \\
\hline
3&(21)& $g_1 g_2 + g_3 $
&$ 3 g_1 g_4 + 2g_2 g_3 + 5g_5 $\\
\hline
4 & (4) & ${1\over 4} (g_1^3 + 3 g_1 g_2 + g_3)$ 
&${5\over 2}g_1^2 g_3 + {5\over 4}g_1g_2^2 + {25\over 4} g_4 + {15\over 4} g_2 g_3 + {15\over 4}g_5$ \\
\hline
4 & (31) & $g_1^2 g_2  +  g_2^2 + 2 g_1 g_3 +g_4 $
& $ g_1(6 g_1 g_4 +8 g_2 g_3 +20 g_5 ) + g_2(g_2^2 +14 g_4)+6 g_3^2 +15 g_6$\\ 
\hline
4 & (22) &${1\over 2}g_1^2 g_2 +{1\over 4}g_2^2 + g_1 g_3   +  {1\over 2}g_4$
& $ g_1( {11\over 4} g_1 g_4 +{13\over 4} g_2  g_3 + {35\over 4} g_5) 
+g_2( {1\over 4}g_2^2 +{19\over 4} g_4 )  +{5\over 2}g_3^2 + {25\over 4}g_6$\\ 
\hline
\end{tabular}}
\end{table}

Further examples, with $\ell(\mu)=3$ and $N=3, 4, 5$ or $6$, are given in Tables \ref{Hurwitz_table6}
and \ref{Hurwitz_table7}.

\begin{table}[H]
\label{Hurwitz_table6}
\small{
\caption{\small{ Nonconnected weighted Hurwitz numbers $H^d_G(\mu)$ with $\ell(\mu)=3$, $ N= 3, 4, 5$ or $6$}}}
\center{
\begin{tabular}{c c c c }\\
$N$ & $\mu$ & $d$ & $H^d_G(\mu)$  \\
\hline
\hline
3&(1,1,1) & $2$ 
& ${1\over 2} g_2 $
\\
\hline
4&(2,1,1)& $3 $
& ${5\over 4} (g_1g_2 +  g_3) $ 
\\
\hline
5&(2,2,1)& $2 $
&${1\over 8} g_1^2$
\\
\hline
5 & (2,2,1) & $4$ 
&$ g_1^2 g_2 +{1\over 4} g_2^2+{7\over 4} g_4$ 
\\
\hline
6 & (3,2,1) & $3 $
& $  {1/6} ( g_1^3 + g_1g_2 )  $
\\ 
\hline
6 & (3,2,1) &$5$ 
&  $ {1/6} (11 g_1^3 g_2+ 31 g_1^2 g_3 +15 g_2g_3 + 
   2g_1(9 g_2^2 + 13 g_4)) +  g_5  $
\\
\hline
\end{tabular}}
\end{table}

\begin{table}[H]
\label{Hurwitz_table7}
\small{
\caption{\small{ Connected weighted Hurwitz numbers $\tilde{H}^d_G(\mu)$ with $\ell(\mu)=3$, $ N=3, 4, 5$ or $6$.
By the Riemann-Hurwitz formula, $\tilde{H}^d_G(1,1,1) = 0$ for $d<4$, $\tilde{H}^d_G(2,1,1) = 0$ for $d<5$, $\tilde{H}^d_G(2,2,1) = 0$ for $d<6$
and $\tilde{H}^d_G(2,2,2) = \tilde{H}^d_G(3,2,1) =0$ for $d<7$. }}
\begin{tabular}{c c c c }\\
$N$ & $\mu$ & $d$ & $\tilde{H}^d_G(\mu)$  \\
\hline
\hline
3&(1,1,1) & $4$ 
& ${1\over 3}(g_2^2 + g_1 g_3 + 2 g_4) $
\\
\hline
4&(2,1,1)& $5 $
& ${1\over 2}(2g_1^2g_3 + 7 g_2 g_3 +g_1(3 g_2^2 + 7 g_4) + 5 g_5) $ 
\\
\hline
4&(2,1,1)& $7 $
& $5 g_2^2 g_3 + 22 g_3 g_4 + 10 g_1^2 g_5 + 
 15 g_2 (g_1 g_4 + 2 g_5) + 5 g_1 (g_3^2 + 8 g_6) + 35 g_7 $ 
\\
\hline
5&(2,2,1)& $6 $
&$g_2^3 + g_1^3 g_3 + 5 g_2 g_4 + g_1^2(2g_2^2 + 5 g_4) + g_1( 9 g_2 g_3 + 7 g_5) +3(g_3^2 +g_6)$
\\
\hline
6 & (2,2,2) & $7$ 
&$ 1/6 (2 g_1^4 g_3 + 11 g_2^2 g_3 + 17 g_3 g_4 + 
   g_1^3 (5 g_2^2 + 13 g_4) + 13 g_2 g_5 $ 
   \\
 &  & $ $ 
&$  +  3 g_1^2 (11 g_2 g_3 + 9 g_5) 
  + g_1 (7 g_2^3 + 22 g_3^2 + 36 g_2 g_4 + 23 g_6) + 7 g_7)$ 
\\ 
\hline
6 & (3,2,1) &$7$ 
&$2 g_1^4 g_3 + 18 g_2^2 g_3+ 17 g_3 g_4 + 
 g_1^3 (5 g_2^2 + 13 g_4) + 18 g_2 g_5  $
\\ 
 &  &$$ 
&$
+ g_1^2 (35 g_2g_3+ 27 g_5) 
+ g_1 (10 g_2^3 + 22 g_3^2 + 43 g_2 g_4 + 23 g_6) + 7 g_7$
\\
\hline
\end{tabular}}
\end{table}

\subsection{Exponential case (simple Hurwitz numbers)}
\label{exponential_case}
Choosing the exponential function $G(z) = e^z$, with Taylor
coefficients
\be
g_i = {1\over i!}
\label{exp_gi}
\ee
as weight generating function corresponds \cite{GH2, H1} to a Dirac measure for the 
weighted Hurwitz numbers supported uniformly on simple branch points,  with ramification profiles 
that are all $2$-cycles:
\be
(\mu^{(1)}, \dots , \mu^{(k)}) = (\underbrace{(2,(1)^{N-2}), \dots , (2,(1)^{N-2})}_{k \text{\ times}}),
\ee
the case considered in \cite{Ok, Pa}. Note that for this case $d=k$ and
\be
H^d_{\exp}(\mu) = {1\over d!} H(\underbrace{(2,(1)^{N-2}), \dots , (2,(1)^{N-2})}_{k \text{\ times}}, \mu).
\ee
Tables \ref{Hurwitz_table8}-\ref{Hurwitz_table9}
give the evaluations of $\tilde{H}^d_{\exp}(\mu)$ and $H^d_{\exp}(\mu)$ for this case.

\begin{table}[H]
\label{Hurwitz_table8}
\small{
\caption{\small{Connected and nonconnected simple Hurwitz numbers $H^d_{\exp}(\mu)$ with $\ell(\mu)=1$ or $2$, \ $N=2,3$ or $4$}}}
\center{
\begin{tabular}{c c c c c c }\\
$N$ & $\mu$ & $\tilde{H}_{\exp}^{N+\ell(\mu)-2}$  & $\tilde{H}_{\exp}^{N+\ell(\mu)}$ & $H_{\exp}^{N+\ell(\mu)-2}$ & $H_{\exp}^{N+\ell(\mu)}$  \\
\hline
\hline
2&(2) & ${1\over 2}$ & ${1\over 12}$ & ${1\over 2} $ & ${1\over 12} $\\
\hline
3&(3)& ${1\over 2}$ & ${3\over 8}$ & ${1\over 2}$ & ${3\over 8} $\\
\hline
3&(21)& ${3\over 4}$ & ${1\over 3}$ & ${3\over 4} $
&$ {27\over 80} $
\\
\hline
4 & (4) & ${2\over 3}$ & ${4\over 3}$ & ${2\over 3} $ &${4\over 3}$  \\
\hline
4 & (31) & ${9\over 8}$ & ${27\over 16}$ & ${3\over 4}$
& $ {9\over 5}$\\ 
\hline
4 & (22) & ${1\over 2}$ & ${2\over 3}$ &${13\over 24}$
& ${121\over 180}$\\ 
\hline
\end{tabular}}
\end{table}

\begin{table}[H]
\label{Hurwitz_table9}
\small{
\caption{\small{ Connected and nonconnected simple Hurwitz numbers $H^d_{\exp}(\mu)$ with $\ell(\mu)=3$, $ N= 3, 4, 5$ or $6$}}}
\center{
\begin{tabular}{c c c c c}\\
$N$ & $\mu$ & $d$ & $\tilde{H}^d_{\exp}(\mu)$ & $H^d_{\exp}(\mu)$  \\
\hline
\hline
$3$ & $(1,1,1)$ & $4$ &  ${1\over 6}$  & $ {3\over 16} $
\\
\hline
$4$&$(2,1,1)$ & $5 $  & $1$ & $ {41\over 30} $ 
\\
\hline
$4$&$(2,1,1)$ & $7 $   & ${13\over 12} $ & ${73\over 63}$
\\
\hline
$5$ & $(2,2,1)$ & $6$  &$ 2$ &${521\over 180}$
\\
\hline
$6$ & $(2,2,2)$ & $7$ &${4\over 3}$  & $ {9853\over 5760} $ 
\\ 
\hline
$6$ & $(3,2,1)$ &$7$  & $ 9$ & ${2511\over 160}$
\\
\hline
\end{tabular}}
\end{table}

\subsection{Rational case}
\label{rational_case}

The Taylor coefficients $\{g_i({\bf c}, {\bf d})\}_{i\in \Nb^+}$  for rational weight generating functions 
$G=G_{{\bf c}, {\bf d}}$, are given  in eq.~(\ref{rational_g_i}) in terms of  the elementary and complete symmetric
polynomials $\{e_i({\bf c})\}$ and  $\{h_i({\bf d})\}$ in the parameters ${\bf c} = (c_1, \dots , c_L)$ 
and ${\bf d} = (d_1, \dots , d_M)$. 
Substituting these in Tables \ref{Hurwitz_table4} - \ref{Hurwitz_table7}, 
we obtain the specialization to rationally weighted Hurwitz numbers.

In particular, the case of weakly monotone (or signed) single Hurwitz numbers ($L=0$, $M=1$, $d_1 =1$)
 is obtained by substituting the values
\be
g_i =1,\quad\forall i \in \Nb^+
\ee
in Tables \ref{Hurwitz_table4} - \ref{Hurwitz_table7}. The case $L=2$, $M=0$, gives the
weighted enumeration of Belyi curves, with three branch points, two weighted ones,
with ramification profiles $(\mu^{(1)}, \mu^{(2)})$ and one unweighted one, with profile $\mu$.
It is obtained by substituting
\be
g_1 = c_1 + c_2, \quad g_2 = c_1 c_2, \quad g_i = 0, \quad \forall i\ge 3
\ee
in Tables \ref{Hurwitz_table4} - \ref{Hurwitz_table7}.


\subsection{Quantum case}
\label{quantum_case}

For the case of the quantum weight generating function $G=H_q$, the Taylor coefficients
are given by eq.~(\ref{g_qqi}).
Substituting these in Tables \ref{Hurwitz_table4} - \ref{Hurwitz_table7}, 
we obtain the specialization to quantum weighted Hurwitz numbers 
given in Tables \ref{Hurwitz_table10}-\ref{Hurwitz_table13}.

\begin{table}[H]
\label{Hurwitz_table10}
\small{
\caption{\small{Nonconnected quantum weighted Hurwitz numbers
 $H^d_{G_{H_q}}(\mu)$ with $\ell(\mu)=1$ or $2$, \ $N=2,3$ or $4$}}
\begin{tabular}{c c c c }\\
$N$ & $\mu$ & $H_{G_{H_q}}^{N+\ell(\mu)-2}$ & $H_{G_{H_q}}^{N+\ell(\mu)}$  \\
\hline
\hline
2&(2) & ${1\over 2(q;q)_1} $ 
& ${1\over 2 (q;q)_3} $\\
\hline
3&(3)& ${2+q\over 3(q;q)_3} $
& ${10 + 5 q + 6 q^2 + 5 q^3 + q^4\over 3(q;q)_4}$ \\
\hline
3&(21)& ${ 5 + 2 q + 2 q^2 \over  2(q;q)_3}$
&$  {21 + 10 q + 14 q^2 + 14 q^3 + 14 q^4 + 4 q^5 + 
 4 q^6 \over 2(q;q)_5} $\\
\hline
4 & (4) & ${5 + 5 q + 5 q^2 + q^3\over 4 (q;q)_3}$ 
&${5 (14 + 14 q + 21 q^2 + 24 q^3 + 25 q^4 + 14 q^5 + 11 q^6 + 4 q^7 + 
   q^8)\over 4 (q;q)_5}$ \\
\hline
4 & (31) & $25 + 20 q + 27 q^2 + 23 q^3 + 10 q^4 + 3 q^5\over 3(q;q)_4 $
& ${{84 + 77 q + 125 q^2 + 156 q^3 + 198 q^4 + 191 q^5 + 163 q^6 \atop+ 
 124 q^7 + 94 q^8 + 52 q^9 + 21 q^{10} + 10 q^{11} + q^{12}} \over (q;q)_6}$\\ 
\hline
4 & (22) &${10 + 10 q + 13 q^2 + 12 q^3 + 5 q^4 + 2 q^5\over4 (q;q)_4}$
& $ {{231 + 231 q + 370 q^2 + 469 q^3 + 589 q^4 + 579 q^5 + 489 q^6 \atop + 
 378 q^7 + 281 q^8 + 161 q^9 + 62 q^{10} + 30 q^{11} + 2 q^{12}}\over 8(q;q)_6}$\\ 
\hline
\end{tabular}}
\end{table}

\begin{table}[H]
\label{Hurwitz_table11}
\small{
\caption{\small{Connected quantum weighted Hurwitz numbers, $\tilde{H}^d_{G_{H_q}}(\mu)$
 with $\ell(\mu)=1$ or $2$, \ $N=2,3$ or $4$}}
\begin{tabular}{c c c c }\\
$N$ & $\mu$ & $\tilde{H}_{G_{H_q}}^{N+\ell(\mu)-2}$ & $\tilde{H}_{G_{H_q}}^{N+\ell(\mu)}$  \\
\hline
\hline
2&(2) & ${1\over 2(q;q)_1}$
 & ${1\over 2 (q;q)_2} $\\
\hline
3&(3)& ${2+q\over 3(q;q)_3}$
& ${10 + 5 q + 6 q^2 + 5 q^3 + q^4\over 3(q;q)_4}$ \\
\hline
3&(21)& ${2 + q + q^2\over (q;q)_3} $
&$ {10 + 5 q + 7 q^2 + 7 q^3 + 7 q^4 + 2 q^5 + 2 q^6\over (q;q)_5}$\\
\hline
4 & (4) & ${5 + 5 q + 5 q^2 + q^3\over 4 (q;q)_3}$ 
&${5 (14 + 14 q + 21 q^2 + 24 q^3 + 25 q^4 + 14 q^5 + 11 q^6 + 4 q^7 + 
   q^8)\over 4 (q;q)_5}$ \\
\hline
4 & (31) & ${5 + 5 q + 7 q^2 + 6 q^3 + 3 q^4 + q^5\over (q;q)_4} $
& $ {{70 + 70 q + 115 q^2 + 145 q^3 + 185 q^4 + 180 q^5 + 156 q^6 \atop+ 
 120 q^7 + 91 q^8 + 51 q^9 + 21 q^{10} + 10 q^{11} + q^{12}} \over (q;q)_6}$\\ 
\hline
4 & (22) &${9 + 9 q + 12 q^2 + 11 q^3 + 5 q^4 + 2 q^5\over 4(q;q)_4}$
& ${{114 + 114 q + 183 q^2 + 232 q^3 + 292 q^4 + 287 q^5 + 243 q^6 \atop
+  188 q^7 + 140 q^8 + 80 q^9 + 31 q^{10} + 15 q^{11} + q^{12}}\over 4(q;q)_6}$\\ 
\hline
\end{tabular}}
\end{table}

\begin{table}[H]
\label{Hurwitz_table12}
\small{
\caption{\small{ Nonconnected quantum weighted Hurwitz numbers $H^d_{G_{H_q}}(\mu)$ 
with $\ell(\mu)=3$, $ N= 4, 5$ or $6$}}}
\center{
\begin{tabular}{c c c c }\\
$N$ & $\mu$ & $d$ & $H^d_{G_{H_q}}$  \\
\hline
\hline
3&(1,1,1) & $2$ 
& ${1\over 2(q;q)_2}$
\\
\hline
4&(2,1,1)& $3 $
& ${5\over 4} {2+q+q^2\over (q;q)_3} $ 
\\
\hline
5&(2,2,1)& $2 $
&${1+q\over 8 (q;q)_2}$
\\
\hline
5 & (2,2,1) & $4$ 
&$ {14 + 16 q + 21 q^2 + 20 q^3 + 9 q^4 + 4 q^5\over 4(q;q)_4}$ 
\\
\hline
6 & (3,2,1) & $3 $
& $  {(2 + q)(1+q+q^2)\over 6 (q;q)_3} $
\\ 
\hline
6 & (3,2,1) &$5$ 
&  $ {107 + 172 q + 287 q^2 + 369 q^3 + 409 q^4 + 319 q^5 + 248 q^6 + 
 133 q^7 + 51 q^8 + 11 q^9 \over 6 (q;q)_5}$
\\
\hline
\end{tabular}}
\end{table}

\begin{table}[H]
\label{Hurwitz_table13}
\small{
\caption{\small{ Connected quantum weighted Hurwitz numbers $\tilde{H}^d_{G_{H_q}}(\mu)$ 
with $\ell(\mu)=3$, $ N= 4, 5$ or $6$.}}}
\center{
\begin{tabular}{c c c c }\\
$N$ & $\mu$ & $d$ & $\tilde{H}^d_{G_{H_q}}$  \\
\hline
\hline
3&(1,1,1) & $4$ 
& ${4 + 2 q + 3 q^2 + 2 q^3 + q^4\over 3 (q;q)_4} $
\\
\hline
4&(2,1,1)& $5 $
& ${24 + 24 q + 39 q^2 + 44 q^3 + 47 q^4 + 28 q^5 + 23 q^6 + 8 q^7 + 3 q^8\over 2(q;q)_5} $ 
\\
\hline
4&(2,1,1)& $7 $
& ${{162 + 162 q + 279 q^2 + 371 q^3 + 518 q^4 + 593 q^5 + 685 q^6 + 
  598 q^7 + 598 q^8  \atop
  + 491 q^9+ 404 q^{10} + 257 q^{11} + 182 q^{12} + 
  90 q^{13} + 50 q^{14} + 15 q^{15} + 5 q^{16}}\over (q;q)_7} $ 
\\
\hline
5&(2,2,1)& $6 $
&${36 + 54 q + 99 q^2 + 141 q^3 + 189 q^4 + 207 q^5 + 204 q^6 + 
 179 q^7
 \atop  + 143 q^8 + 97 q^9 + 53 q^{10 }+ 28 q^{11} + 8 q^{12} + 2 q^{13}}\over (q;q)_6$
\\
\hline
6 & (2,2,2) & $7$ 
&$ {{216 + 432 q + 891 q^2 + 1485 q^3 + 2295 q^4 + 3099 q^5 + 3922 q^6 + 
 4377 q^7 + 4689 q^8 + 4567 q^9  + 4157 q^{10}  \atop
+ 3435 q^{11} + 2655 q^{12} + 
 1837 q^{13} + 1173 q^{14} + 638 q^{15} + 301 q^{16} + 117 q^{17} + 29 q^{18} +  5 q^{19}}
 \over 6 (q;q)_7}$ 
\\ 
\hline
6 & (3,2,1) &$7$ 
&${{240 + 480 q + 1002 q^2 + 1664 q^3 + 2584 q^4 + 3484 q^5 + 4416 q^6 + 
 4927 q^7 + 5288 q^8 + 5139 q^9 + 4687 q^{10} \atop
 + 3863 q^{11} + 2990 q^{12} +  2063 q^{13} + 1319 q^{14} + 711 q^{15} + 338 q^{16} + 128 q^{17} + 32 q^{18} +  5 q^{19}} \over (q;q)_7}$
\\
\hline
\end{tabular}}
\end{table}

 \bigskip
\noindent 
\small{ {\it Acknowledgements.} 
This work was partially supported by the Natural Sciences and Engineering Research Council of Canada (NSERC) 
and the Fonds de recherche du Qu\'ebec, Nature et technologies (FRQNT).  
\bigskip
\bigskip 

 
 \newcommand{\arxiv}[1]{\href{http://arxiv.org/abs/#1}{arXiv:{#1}}}

\bigskip
\noindent


\begin{thebibliography}{99}

\bibitem{AC1} J. Ambj{\o}rn and L. Chekhov, ``The matrix model for dessins d'enfants'', 
{\em Ann. Inst. Henri Poincar\'e, Comb. Phys. Interact.} {\bf  1},  337-361 (2014).

\bibitem{AC2} J. Ambjorn and L. Chekhov, ``A matrix model for hypergeometric Hurwitz numbers'',
{\em Theor. Math. Phys.'} {\bf 181}, 1486-1498 (2014).

 \bibitem{ACEH1} A.~Alexandrov, G. Chapuy, B. Eynard and J. Harnad, 
``Weighted Hurwitz numbers and topological recursion: an overview'', 
{\em J. Math. Phys.} {\bf 59},  081102: 1-20 (2018).

\bibitem{ACEH2} A.~Alexandrov, G.~Chapuy, B.~Eynard and J.~Harnad,
  ``Fermionic approach to weighted Hurwitz numbers and topological recursion,''
  {\em Commun.\ Math.\ Phys.} {\bf 360} no.2,  777  (2018).
  
  \bibitem{ACEH3} A.~Alexandrov, G.~Chapuy, B.~Eynard and J.~Harnad,
  ``Weighted Hurwitz numbers and topological recursion''
    \arxiv 1806.09738 (to appear in   {\em Commun.\ Math.\ Phys.} (2020)).

  \bibitem{AMMN}  A.~Alexandrov, A.~Mironov, A. Morozov  and S. Natanzon, ``
On KP-integrable Hurwitz functions'',  {\em JHEP} {\bf 1411} 080  (2014).

   \bibitem{Be} Dan Bernstein, ``The computational complexity of rules for the character table of $S_n$'', 
   {\em J. of Symb. Comp.} {\bf 37}, 727-748 (2014). 
 
    \bibitem{BEMS} G. Borot, B. Eynard, M. Mulase and B. Safnuk, ``A matrix model for simple Hurwitz numbers, and topological recursion'', {\em J. Geom. Phys.} {\bf 61}, 522--540 (2011). 
    
 \bibitem{BH}  M. Bertola and  J. Harnad, ``Rationally weighted Hurwitz numbers, Meijer G-functions and matrix integrals'', arXiv:1904.03770,
 {\em J. Math. Phys.} (in press, 2019).
 
\bibitem{EO2} B. Eynard  and N. Orantin, ``Topological recursion in enumerative geometry and random matrices'', 
{\em J. Phys. A} {\bf 42 },  293001 (2009).

\bibitem{Frob1}  G. Frobenius, ``\"Uber die Charaktere der symmetrischen Gruppe'', {\em Sitzber. Pruess. Akad. Berlin}, 516-534  (1900).

\bibitem{Frob2}  G. Frobenius, ``\"Uber die Charakterische Einheiten der symmetrischen Gruppe'', {\em Sitzber.  Akad. Wiss., Berlin}, 328-358  (1903). Gesammelte Abhandlung III, 244-274.

\bibitem{FH} William Fulton and Joe Harris,  ``Representation Theory",  Graduate Texts in Mathematics, {\bf 129}, 
Springer-Verlag (N.Y., Berlin, Heidelberg, 1991), Chapt. 4, and Appendix A.

\bibitem{GGN1} I. P. Goulden, M. Guay-Paquet and J. Novak, ``Monotone Hurwitz numbers and the HCIZ Integral'', 
 {\em Ann. Math. Blaise Pascal} {\bf 21}, 71-99 (2014).

\bibitem{GH1} M. Guay-Paquet and J. Harnad, ``2D Toda $\tau$-functions as combinatorial generating functions'',  {\em Lett. Math. Phys.} {\bf 105}, 827-852 (2015).

\bibitem{GH2} M. Guay-Paquet and J. Harnad, ``Generating functions for weighted Hurwitz numbers'',  {\em J. Math. Phys.} {\bf  58}, 083503 (2017).

\bibitem{HC} Harish-Chandra,  ``Differential operators on a semisimple Lie algebra'',
{\em Amer. J. Math.} {\bf 79}, 87-120 (1957).

\bibitem{H1} J. Harnad,  ``Weighted Hurwitz numbers and hypergeometric $\tau$-functions: an overview'', {\em  Proc.~Symp.~Pure Math} {\bf 93}, 289-333 (2016). 

\bibitem{H2} J. Harnad,  ``Quantum Hurwitz numbers  and Macdonald polynomials'', {\em J. Math. Phys.} {\bf 57} 113505 (2016).

\bibitem{H3} J. Harnad, ``Multispecies weighted Hurwitz numbers'', {\em SIGMA} {\bf 11} 097, (2015).

\bibitem{HO} J. Harnad and A. Yu. Orlov, ``Hypergeometric $\tau$-functions, Hurwitz numbers and enumeration of paths'',  {\em Commun. Math. Phys. } {\bf 338}, 267-284 (2015).

\bibitem{Hu1} A. Hurwitz, ``\"Uber Riemann'sche Fl\"asche mit gegebnise Verzweigungspunkten'', 
{\em Math. Ann.} {\bf 39}, 1-61 (1891); Matematische Werke I, 321-384.

\bibitem{Hu2} A. Hurwitz, ``\"Uber die Anzahl der Riemann'sche Fl\"asche mit gegebnise Verzweigungspunkten'', 
{\em Math. Ann.} {\bf 55}, 53-66 (1902); Matematische Werke I, 42-505.

\bibitem{IZ}  C. Itzykson C and J.-B. Zuber  ``The planar approximation. II '', {\em J. Math. Phys.} {\bf  21},
 411-421 (1980 ).

\bibitem{KZ} M. Kazarian and P. Zograf, ``Virasoro constraints and topological recursion for Grothendieck's dessin counting'',  {\it Lett. Math. Phys.}  {\bf 105} 1057-1084 (2015).

\bibitem{LZ} S. K.  Lando and A.K.  Zvonkin, {\em Graphs on Surfaces and their Applications}, Encyclopaedia of Mathematical Sciences, Volume {\bf 141}, Springer-Verlag (2004).
 
\bibitem{Mac} I.~G. ~Macdonald, {\em Symmetric Functions and Hall Polynomials},
Clarendon Press, Oxford, (1995).

\bibitem{Ok} A.~Okounkov, ``Toda equations for Hurwitz numbers'', {\em Math.~Res.~Lett.} {\bf 7}, 447--453 (2000).
   
\bibitem{OrSc} A.~Yu.~Orlov and D.~M.~Scherbin, ``Hypergeometric solutions of soliton equations'', {\em Theor. Math. Phys.} {\bf 128}, 906-926 (2001).

\bibitem{Pa} R. Pandharipande, ``The Toda Equations and the Gromov-Witten Theory of
the Riemann Sphere'',  {\em Lett. Math. Phys.} {\bf  53}, 59-74 (2000).

\bibitem{Sa} M. Sato, ``Soliton equations as dynamical systems on infinite dimensional Grassmann 
manifolds'', RIMS, Kyoto Univ. {\it Kokyuroku} {\bf 439},  30-46 (1981).

\bibitem{Sch} I. Schur, ``Neue Begr\"undung' der Theorie der Gruppencharaktere'', {\em Sitzber.  Akad. Wiss., Berlin}, 406-432  (1905).

 \bibitem{SW} G. Segal  and G. Wilson, ``Loop groups and equations of KdV type '', {\em Publ. Math. IH\'ES} {\bf 6}, 5-65 (1985).

\bibitem{Z} P. Zograf,  ``Enumeration of Grothendieck's dessins and KP hierarchy'', 
{\em Int. Math Res. Notices} {\bf 24}, 13533-13544 (2015).

\end{thebibliography}
\end{document}